\documentclass[11pt]{article}

\usepackage{amsfonts, amsmath, amssymb, amsthm}
\usepackage{a4}

\theoremstyle{plain}
\newtheorem{theorem}{Theorem}
\newtheorem{lemma}{Lemma}

\theoremstyle{definition}
\newtheorem{dfn}{Definition}
\newtheorem{xmp}{Example}

\theoremstyle{remark}
\newtheorem*{remark}{Remark}

\DeclareMathOperator{\const}{const}
\DeclareMathOperator{\rank}{rank}
\DeclareMathOperator{\sn}{sn}
\DeclareMathOperator{\cn}{cn}
\DeclareMathOperator{\dn}{dn}
\DeclareMathOperator{\diag}{diag}

\author{Igor G. Korepanov}
\title{Two-cocycles give a full nonlinear parameterization of the simplest 3--3 relation}
\date{October 2013 --- May 2014}

\begin{document}

\maketitle

\begin{abstract}
A parameterization of Grassmann-algebraic relations corresponding to the Pachner move 3--3 is proposed. In these relations, each 4-simplex is assigned a Grassmann weight depending on five anticommuting variables associated with its 3-faces. The weights are chosen to have the ``simplest'' form --- a Grassmann--Gaussian exponent or its analogue (satisfying a similar system of differential equations). Our parameterization works for a Zariski open set of such relations, looks relevant from the algebraic-topological viewpoint, and reveals intriguing \emph{nonlinear} relations between objects associated with simplices of different dimensions.
\end{abstract}


\medskip


\bigskip

\section{Introduction}\label{s:i}

\subsubsection*{Pachner moves --- elementary local rebuildings of a manifold triangulation}

In order to construct a topological field theory for piecewise linear (PL) manifolds, it makes sense first to construct algebraic relations corresponding to \emph{Pachner moves}, as is explained, for instance, in~\cite[Section~1]{Lickorish}. Due to the Pachner's theorem stating that a triangulation of a PL manifold can be transformed into any other triangulation using a finite sequence of these moves~\cite{Pachner}, there is then a hope that some quantities can be derived from such algebra characterizing the whole manifold.

Here we will be dealing with the four-dimensional case, so we recall what four-dimensional Pachner moves are. Each of them replaces a cluster of 4-simplices, which we call the \emph{left-hand side} (l.h.s.) of the move, with a cluster of some other 4-simplices --- its \emph{right-hand side} (r.h.s.), occupying the same place in the triangulation and having the same boundary. Gluing the withdrawn and the replacing clusters together (using the identity mapping of the boundary, and forgetting for a moment about the rest of the manifold), one must get a sphere~$S^4$ triangulated in five 4-simplices as the boundary~$\partial \Delta^5$ of a 5-simplex~$\Delta^5$. A pedagogical introduction to Pachner moves can be found in~\cite{Lickorish}; in the present paper, we will be dealing only with the move 3--3 (replacing three 4-simplices with three other ones) which is, in some informal sense, central: experience shows that if we have managed to find an algebraic formula whose structure can be regarded as reflecting the structure of this move, then we can also find (usually more complicated) formulas corresponding to the rest of Pachner moves.

\subsubsection*{Grassmann--Berezin calculus of anticommuting variables}

The simplest nontrivial relations corresponding to Pachner moves seem to arise in \emph{Grassmann algebras}. The Grassmann--Berezin calculus~\cite{B,B-super} of anticommuting variables appeared naturally in the author's work~\cite{A-style} as a means for elegant formulation of what happens with some invariants related to ``exotic Reidemeister torsions'' when gluing two manifolds together; the paper~\cite{A-style} may be compared with the preceding paper~\cite{torsions} where Grassmann--Berezin calculus was not yet used.

In this paper, we continue the work begun in~\cite{2-cocycles,KS2}. We take the simplest possible form~\eqref{33} of Grassmann-algebraic relation corresponding to Pachner move 3--3, where just one Grassmann variable is attached to a 3-face. We further assume that the Grassmann weight of a 4-simplex has a special form, depending on the five variables on its 3-faces and obeying a certain system of differential equations. One typical case of such weight is a Grassmann--Gaussian exponent~\eqref{xFx}.

How ``local'' relations in Grassmann algebra corresponding, in some informal sense, to Pachner moves can lead to ``global'' manifold invariants, has been shown already in paper~\cite{A-style}. We would like also to refer to the experience gained in the area of Yang--Baxter equation and Zamolodchikov tetrahedron equation in mathematical physics, where every object possessing a pictorial representation like those typical for the mentioned equations, turned out eventually to be relevant for the theory.

\subsubsection*{The results of this paper}

We give a full parameterization for our 3--3 relations, that is, embracing a Zariski open subset of their 18-parametric family. The key component in our parameterization is an arbitrary 2-cocycle (over the field~$\mathbb C$) on~$\partial\Delta^5$; there is also a ``gauge freedom'' \eqref{renorm},~\eqref{interc}. Our parameterization looks relevant from the viewpoint of extending it to other Pachner moves, and to entire manifolds, and this is what was lacking in our previous parameterization already proposed in~\cite[Section~5]{KS2}. 

One more result is some very simple formulas (namely, \eqref{1:c} and~\eqref{e}) using elliptic functions and pertaining to a \emph{single} 4-simplex. These intriguing formulas give rise to the idea that there may also be interesting objects of \emph{algebraic-geometrical} nature pertaining to the whole 3--3 move, and awaiting their discovery.

\subsubsection*{Edge operators and a nonlinear algebraic topology}

Interesting nonlinear algebraic relations appear already within one 4-simplex equipped with a Grassmann--Gaussian exponent~\eqref{xFx}. The key idea is to introduce \emph{edge operators} --- first order differential operators annihilating this exponent (or its analogue) and involving only Grassmann variables in tetrahedra containing a given edge. Remarkable properties of edge operators, found with the aid of computer algebra, lead to the existence of a 2-cocycle parameterizing the Grassmann--Gaussian exponents (or their analogues) to within the gauge freedom mentioned above. And, moreover, edge operators compose in a nice way when 4-simplices are glued together, if a 2-cocycle is given on their union. This is what the present work is based upon.

\subsubsection*{Organization of the paper}

Below,
\begin{itemize}\itemsep 0pt
 \item in Section~\ref{s:GBC}, we introduce our tools: Grassmann--Berezin calculus of anticommuting variables and complex Euclidean spaces, and remind some basic facts related to them,
 \item in Section~\ref{s:single}, we introduce a \emph{quasi-Gaussian} Grassmann weight for a 4-simplex and show how it brings about a 2-cocycle --- an unexpected result of direct calculations,
 \item in Section~\ref{s:pw}, we study further the relations between a quasi-Gaussian weight and its cocycle. It turns out that 2-cocycles parameterize quasi-Gaussian weights to within some natural gauge transformations,
 \item and in Section~\ref{s:33}, we demonstrate how 2-cocycles given on the boundary of a 5-simplex give rise to (a Zariski open subset of all) relations with quasi-Gaussian weights corresponding to a Pachner move 3--3.
\end{itemize}

\section{Some basic facts related to Grassmann algebras}\label{s:GBC}

\subsection{Grassmann algebras and Berezin integral}\label{ss:GB}

\begin{dfn}\label{dfn:GA}
In this paper, a \emph{Grassmann algebra} is an associative algebra over the field~$\mathbb C$ of complex numbers, with unity, generators~$x_i$ --- also called Grassmann variables --- and relations
\begin{equation*}
x_i x_j = -x_j x_i .
\end{equation*}
\end{dfn}

This implies that, in particular, $x_i^2 =0$, so each element of a Grassmann algebra is a polynomial of degree $\le 1$ in each~$x_i$.

\begin{remark}
Grassmann algebra is also known as the \emph{exterior algebra} of the vector space generated by~$x_i$.
\end{remark}

The \emph{degree} of a Grassmann monomial is its total degree in all Grassmann variables. If an algebra element consists of monomials of only odd or only even degrees, it is called \emph{odd} or, respectively, \emph{even}. If all the monomials have degree~2, we call such element a Grassmannian \emph{quadratic form}.

The \emph{exponent} is defined by its usual Taylor series.

\begin{dfn}\label{dfn:GGe}
A \emph{Grassmann--Gaussian exponent} is the exponent of a quadratic form.
\end{dfn}

\begin{xmp}
For $\lambda,\mu,\nu\in\mathbb C$, the following is an example of Grassmann--Gaussian exponent:
\begin{equation*}
\exp (\lambda x_1x_2+\mu x_2x_3+\nu x_3x_4) = 1+\lambda x_1x_2+\mu x_2x_3+\nu x_3x_4+\lambda \nu x_1x_2x_3x_4 .
\end{equation*}\end{xmp}

The concept of \emph{derivative} is extended onto Grassmann algebras, due to their noncommutativity, in two ways.

\begin{dfn}\label{dfn:der}
The \emph{left derivative} and \emph{right derivative} with respect to a Grassmann variable~$x_i$ are $\mathbb C$-linear operations in Grassmann algebra, denoted $\dfrac{\partial}{\partial x_i}$ and~$\dfrac{\overleftarrow{\partial}}{\partial x_i}$ and defined as follows. Let $f$ be an element not containing variable~$x_i$, then
\begin{equation}\label{Gd1}
\dfrac{\partial}{\partial x_i}f=f\dfrac{\overleftarrow{\partial}}{\partial x_i}=0,
\end{equation}
and
\begin{equation}\label{Gd2}
\dfrac{\partial}{\partial x_i}(x_if)=f,\qquad (fx_i)\dfrac{\overleftarrow{\partial}}{\partial x_i}=f.
\end{equation}
\end{dfn}

\emph{Leibniz rules} for differentiating a product follow directly from \eqref{Gd1} and~\eqref{Gd2} and can be formulated as follows: for $f$ either even or odd,
\begin{equation}\label{L}
\dfrac{\partial}{\partial x_i}(fg) = \dfrac{\partial}{\partial x_i}f\cdot g + \epsilon f\dfrac{\partial}{\partial x_i}g,\qquad (gf)\dfrac{\overleftarrow{\partial}}{\partial x_i}=g\cdot f\dfrac{\overleftarrow{\partial}}{\partial x_i} + \epsilon g\dfrac{\overleftarrow{\partial}}{\partial x_i}f,
\end{equation}
where $\epsilon=1$ for an even~$f$ and $\epsilon=-1$ for an odd~$f$.

A peculiarity of the \emph{Grassmann--Berezin calculus} of anticommuting variables is that the most natural idea of how to define an integral in a Grassmann algebra leads to the same operation as the right derivative. Still, this operation deserves its second name --- Berezin integral, as well as a separate definition, in such contexts where it appears as an analogue of integration, and not differentiation, in the usual calculus.

\begin{dfn}\label{dfn:Bi}
\emph{Berezin integral} in a variable~$x_i$ is a $\mathbb C$-linear operator
\[
f\mapsto \int f\, \mathrm dx_i
\]
in a Grassmann algebra, satisfying
\begin{equation*}
\int \mathrm dx_i =0, \qquad \int x_i\, \mathrm dx_i =1, \qquad \int gh\, \mathrm dx_i = g \int h\, \mathrm dx_i,
\end{equation*}
where $g$ does not contain~$x_i$.

\emph{Multiple integral} is defined according to the following Fubini rule:
\[
\idotsint f\, \mathrm dx_1 \, \mathrm dx_2 \,\dots \, \mathrm dx_n = \int \left( \dots \int \left( \int f \, \mathrm dx_1 \right) \mathrm dx_2 \dots \right) \mathrm dx_n \, .
\]
\end{dfn}

\subsection{Clifford algebra generated by differentiations w.r.t.\ and multiplications by Grassmann variables}\label{ss:C}

Consider a $\mathbb C$-linear combination of operators of left differentiations and left multiplications by Grassmann variables:
\begin{equation}\label{bg}
d = \sum_{t=1}^n (\beta_t\partial_t + \gamma_tx_t),
\end{equation}
where we use a shorthand notation $\partial_t\stackrel{\rm def}{=}\partial / \partial x_t$.

We regard the anticommutator of two operators~\eqref{bg} (defined as $[A,B]_+=AB+BA$ for operators $A$ and~$B$) as their \emph{scalar product}:
\begin{equation}\label{sc}
\langle d^{(1)},d^{(2)}\rangle \stackrel{\rm def}{=} [d^{(1)},d^{(2)}]_+ = \sum_{t=1}^n (\beta_t^{(1)}\gamma_t^{(2)} + \beta_t^{(2)}\gamma_t^{(1)} ).
\end{equation}
With this scalar product, operators~\eqref{bg} form a \emph{complex Euclidean space}, we denote it~$\mathcal V$, while all polynomials of these operators form a \emph{Clifford algebra}.

\begin{remark}
Much interesting material about Clifford algebras (and maximal iso\-tropic spaces in complex Euclidean spaces that appear in our Subsection~\ref{ss:max-iso}) can be found in the book~\cite{C}.
\end{remark}

Separate summands in~\eqref{sc} will be also of use for us, so we introduce the notation
\begin{equation}\label{sct}
\langle d^{(1)},d^{(2)}\rangle_t \stackrel{\rm def}{=} \beta_t^{(1)}\gamma_t^{(2)} + \beta_t^{(2)}\gamma_t^{(1)} .
\end{equation}
Also, we denote~$\mathcal V_t$ the two-dimensional linear space spanned by $\partial_t$ and~$x_t$; it is clear that
\begin{equation*}
\mathcal V = \bigoplus_{t=1}^n \mathcal V_t
\end{equation*}
in the sense of complex Euclidean spaces.

\subsection{Maximal isotropic spaces of operators}\label{ss:max-iso}

\begin{dfn}\label{dfn:iso}
Subspace~$V$ of a complex Euclidean space is called \emph{isotropic} if the scalar product restricted onto~$V$ identically vanishes.
\end{dfn}

Especially interesting for us are maximal --- in the sense of inclusion --- isotropic subspaces.

\begin{xmp}\label{xmp:dx}
Subspace of the space~$\mathcal V$ of operators~\eqref{bg} spanned either by all~$\partial_t$ or all~$x_t$ is maximal isotropic.
\end{xmp}

\begin{theorem}\label{th:eo}
Let $x_t$, $t=1,\ldots,n$ be generators of a Grassmann algebra over~$\mathbb C$, and $\mathcal V$ the $2n$-dimensional complex Euclidean space of operators~\eqref{bg}. Let also $\mathsf p=\begin{pmatrix}\partial_1 & \dots & \partial_n\end{pmatrix}^{\mathrm T}$ and $\mathsf x=\begin{pmatrix}x_1 & \dots & x_n\end{pmatrix}^{\mathrm T}$ be the columns of partial derivatives and corresponding variables. Then,
\begin{enumerate}\itemsep 0pt
\item \label{i:k} the set of all maximal isotropic subspaces $V \subset \mathcal V$ --- called \emph{isotropic Grassmannian} --- splits into two connected components, called below ``even'' and ``odd'',
\item \label{i:l} for a skew-symmetric matrix~$F$, the span of the elements of column $\mathsf p + F \mathsf x$ is a maximal isotropic subspace,
\item \label{i:m} the mapping
\begin{equation}\label{pFx}
F \mapsto \bigl(\text{span of the elements of column }(\mathsf p + F \mathsf x)\bigr)
\end{equation}
is a bijection from the set of all $n\times n$ skew-symmetric matrices~$F$ onto a Zariski open set in the ``even'' maximal isotropic subspace,
\item \label{i:n} choose now a subset $T\subset \{1,\dots,n\}$ of numbers from~$1$ through~$n$, and make, for every $t\in T$, the interchange of two elements $\partial_t \leftrightarrow x_t$ between $\mathsf p$ and~$\mathsf x$. If $T$ is of \emph{even} cardinality, then \eqref{pFx}, with thus modified $\mathsf p$ and~$\mathsf x$, is again a bijection from the set of (all $n\times n$ skew-symmetric) matrices~$F$ onto a Zariski open set in the same ``even'' subspace as in item~\ref{i:m},
\item \label{i:o} if $T$ is of \emph{odd} cardinality, then \eqref{pFx}, with $\mathsf p$ and~$\mathsf x$ modified correspondingly, is a bijection from the set of matrices~$F$ onto a Zariski open set in the ``odd'' maximal isotropic subspace.
\end{enumerate}
\end{theorem}

\begin{proof}
\begin{itemize}\itemsep 0pt
 \item[\ref{i:k}] That isotropic Grassmannian splits up in two connected components, is a known fact, see, e.g., \cite[Theorem~1]{KS2} for a simple proof.
 \item[\ref{i:l}] It is easily checked that, indeed, the scalar product of any two (coinciding or not) entries in the column $\mathsf p + F \mathsf x$ vanishes, and that these~$n$ entries are linearly independent.
 \item[\ref{i:m}] This follows from the obvious injectiveness of mapping~\eqref{pFx} and the fact that the (complex) dimensionality of both the space of matrices~$F$ and the isotropic Grassmannian is the same, namely $n(n-1)/2$. This dimensionality for the isotropic Grassmannian is again a known fact, for which we can again refer to the proof of \cite[Theorem~1]{KS2}, where isotropic subspaces~$V$ are parameterized in terms of orthogonal matrices, see~\cite[formula~(12)]{KS2}, at least in a neighborhood of any given~$V$.
 \item[\ref{i:n}] Assume, for simplicity, that $T=\{1,\dots,k\}$, with even $k\le n$, and write~$F$ in a block matrix form:
\[
F = \begin{pmatrix} A & B \\ -B^{\mathrm T} & C \end{pmatrix},
\]
with skew-symmetric $A$ and~$C$, and $A$ of sizes $k\times k$. Them, a small exercise shows that the same isotropic space is obtained after the interchanges $\partial_t \leftrightarrow x_t$, $t\in T$, provided $F$ is replaced with
\[
F' = \begin{pmatrix} A^{-1} & A^{-1}B \\ -B^{\mathrm T}A^{-1} & B^{\mathrm T}A^{-1}B+C \end{pmatrix}.
\]
As a generic \emph{even-dimensional} skew-symmetric matrix~$A$ is invertible, this shows that we get, indeed, into the same connected component of isotropic Grassmannian.
 \item[\ref{i:o}] Similar argument shows that all subsets~$T$ of odd cardinality lead to one and the same connected component of  isotropic Grassmannian. That this component is different from the ``even'' case, can be shown as follows. Any $n$-dimensional subspace in the space of $2n$-row vectors is the span of the rows of some matrix $\begin{pmatrix} F_1 & F_2 \end{pmatrix}$, where both blocks $F_1$ and~$F_2$ are of sizes $n\times n$. These blocks are determined to within a left multiplication by an arbitrary invertible matrix. Hence, in particular, the expression
\begin{equation}\label{F1F2}
(F_1-F_2)^{-1}(F_1+F_2)
\end{equation}
is an \emph{invariant} of the subspace as such.

Identifying a $2n$-row vector~$\begin{pmatrix} a_1 & \dots & a_n & b_1 & \dots & b_n \end{pmatrix}$ with the operator $a_1\partial_1 + \dots + a_n\partial_n + b_1 x_1 + \dots + b_n x_n$, we see that, for the span of column $\mathsf p + F \mathsf x$ entries, the expression~\eqref{F1F2} gives the identity matrix:
\[
(\mathbf 1_n - F)^{-1}(\mathbf 1_n + F )= \mathbf 1_n ,
\]
and can easily be shown to remain the same if we make an even number of changes $\partial_t \leftrightarrow x_t$ in all these entries. One can check, however, that for an odd number of changes, the expression~\eqref{F1F2} changes to~$-\mathbf 1_n$.
\end{itemize}
\end{proof}

\begin{theorem}\label{th:v}
Let $x_t$, $t=1,\ldots,n$ be generators of a Grassmann algebra over~$\mathbb C$, and $V \subset \mathcal V$ a maximal isotropic subspace in the space of operators~\eqref{bg}. Then,
\begin{enumerate}\itemsep 0pt
 \item\label{i:t} the nullspace of~$V$ (${}={}$the space of vectors annihilated by all elements in~$V$) is one-dimensional,
 \item\label{i:u} if $V$ is a subspace of the kind dealt with in items \ref{i:l} and~\ref{i:m} of Theorem~\ref{th:eo}, i.e.,
  \begin{equation}\label{Vs}
V = \bigl(\text{span of the elements of column }(\mathsf p + F \mathsf x)\bigr),
  \end{equation}
then the nullspace of\/~$V$ is spanned by the Grassmann--Gaussian exponent
  \begin{equation}\label{xFx}
  \mathcal W = \exp \left(-\frac{1}{2}\,\mathsf x^{\mathrm T} F \mathsf x \right),
  \end{equation}
 \item\label{i:1:v} after each interchange $\partial_t \leftrightarrow x_t$ between $\mathsf p$ and~$\mathsf x$ (like in items \ref{i:n} or~\ref{i:o} of Theorem~\ref{th:eo}), an element~$\mathcal W_{\mathrm{new}}$ spanning the nullspace of the new~$V$ can be obtained from an element~$\mathcal W_{\mathrm{old}}$ spanning the nullspace of the old~$V$ as:
  \begin{equation}\label{px}
\mathcal W_{\mathrm{new}} = (\partial_t-x_t)\mathcal W_{\mathrm{old}}.
  \end{equation}
\end{enumerate}
\end{theorem}

\begin{remark}
Of course, the operator in~\eqref{px} is involutive up to a numeric factor: $(\partial_t-x_t)^2=-1$.
\end{remark}

\begin{proof}[Proof of Theorem~\ref{th:v}]
Items \ref{i:t} and~\ref{i:u} make an easy variation on the theme of~\cite[Theorem~2]{KS2}. Item~\ref{i:1:v} is checked by a direct calculation.
\end{proof}

\section{Edge operators in a single 4-simplex}\label{s:single}

In this Section and in the next Section~\ref{s:pw}, we deal with just one 4-simplex~$12345$ (where $1,\ldots, 5$ are its vertices). We define a specific Grassmann weight for it as a Grassmann algebra element annihilated by a maximal isotropic subspace of operators, and we hope to demonstrate that our weight gives rise, already for a single 4-simplex, to an interesting \emph{nonlinear} algebraic topology.

\subsection{The 4-simplex weight}\label{ss:W}

We put a Grassmann variable~$x_t$ in correspondence to each of its five 3-faces --- tetrahedra
\begin{equation}\label{t}
t=2345,\; 1345,\; 1245,\; 1235\;\; \text{and} \;\; 1234.
\end{equation}
Then we proceed along the lines of Subsections \ref{ss:C} and~\ref{ss:max-iso}: consider the ten-dimensional space~$\mathcal V$ of all $\mathbb C$-linear combinations
\begin{equation}\label{bgt}
d = \sum_{t\subset 12345} (\beta_t \partial_t+\gamma_t x_t),
\end{equation}
define the scalar product for operators~\eqref{bgt} as their anticommutator, and choose a maximal isotropic subspace $V\subset \mathcal V$.

\begin{dfn}\label{dfn:W}
A nonzero element of the Grassmann algebra generated by the five~$x_t$ corresponding to the 3-faces of the 4-simplex, annihilated by a given maximal isotropic subspace $V\subset \mathcal V$, is called the 4-simplex \emph{quasi-Gaussian weight} (corresponding to~$V$).
\end{dfn}

In particular, in the ``even'' case of Theorem~\ref{th:eo}, almost all such 4-simplex weights are, up to a factor, Grassmann--Gaussian exponents that can be described as follows. Put a quantity~$\varphi_s \in \mathbb C$ in correspondence to each of the ten 2-faces~$s=ijk$, \ $1\le i<j<k\le 5$, of 4-simplex~$12345$. Then put
\begin{equation}\label{gg}
\mathcal W = \exp \Phi ,
\end{equation}
where $\Phi$ is the following Grassmannian quadratic form:
\begin{equation}\label{Phi}
\Phi = \sum_{\substack{\text{over 2-faces }ijk\\[.3ex] \text{ of }12345}} \epsilon_{lijkm} \, \varphi_{ijk} \, x_{\{ijkl\}} x_{\{ijkm\}}.
\end{equation}
In~\eqref{Phi},
\begin{itemize}\itemsep 0pt
 \item $l<m$ are the two vertices of~$12345$ that do \emph{not} enter in~$ijk$,
 \item $\epsilon_{lijkm}$ is the sign of permutation between the sequence in its subscripts and~$12345$, and
 \item the curly brackets mean that the numbers within them should be put in the increasing order, to represent a tetrahedron in~\eqref{t}, e.g., $\{2341\}=1234$.
\end{itemize} 

The exponent \eqref{gg},~\eqref{Phi} is characterized, up to a factor that does not depend on those~$x_t$ that enter in it, and in full accordance with formulas \eqref{pFx} and~\eqref{xFx}, as follows. Consider the following 5-columns of operators of left differentiations or left multiplications:
\begin{gather}
\mathsf p = \begin{pmatrix} \partial_{2345} & \partial_{1345} & \partial_{1245} & \partial_{1235} & \partial_{1234} \end{pmatrix}^{\mathrm T} , \label{p12345} \\ 
\mathsf x = \begin{pmatrix} x_{2345} & x_{1345} & x_{1245} & x_{1235} & x_{1234} \end{pmatrix}^{\mathrm T} , \label{x12345}
\end{gather}
and the following $5\times 5$ matrix:
\begin{equation}\label{F}
F = \begin{pmatrix}
            0 &  -\varphi_{345} &  \varphi_{245} & -\varphi_{235} &  \varphi_{234} \\
            \varphi_{345} & 0 &  -\varphi_{145} &  \varphi_{135} & -\varphi_{134}  \\
            -\varphi_{245} & \varphi_{145} &  0 & -\varphi_{125} &  \varphi_{124} \\
            \varphi_{235} & -\varphi_{135} &  \varphi_{125} & 0 &  -\varphi_{123}  \\
            -\varphi_{234} &  \varphi_{134} & -\varphi_{124} &  \varphi_{123} & 0 
\end{pmatrix} .
\end{equation}
Then, $\mathcal W$ spans the nullspace of the elements of column given by~\eqref{pFx}.

\subsection{Edge operators}\label{ss:eo}

Especially interesting operators in the space~$V$ are the following \emph{edge operators} that involve only the tetrahedra containing a given edge~$b$.

\begin{dfn}\label{dfn:edge}
An \emph{edge operator} for an edge $b=ij$ is any element in~$V$ having zero coefficients at $x_t$ and~$\partial_t$ if $t\not\supset b$.
\end{dfn}

For a given~$b$ and a \emph{generic}~$V$, the so defined edge operators are easily shown to form a one-dimensional linear space. This is illustrated below in Example~\ref{xmp:e}. We use notation~$d_b$ for an edge operator spanning this one-dimensional space; at this moment, $d_b$ is defined up to a numeric factor, but we will fix its normalization after Theorem~\ref{th:1b}.

\begin{xmp}\label{xmp:e}
Let $V$ be defined as in~\eqref{Vs}, with $\mathsf p$, $\mathsf x$ and~$F$ as in \eqref{p12345}, \eqref{x12345} and~\eqref{F}. Then, any edge operator can be written as
\begin{equation}\label{apx}
\alpha (\mathsf p + F\mathsf x),
\end{equation}
where
\[
\alpha = \begin{pmatrix} \alpha_{2345} & \alpha_{1345} & \alpha_{1245} & \alpha_{1235} & \alpha_{1234} \end{pmatrix}
\]
is a row of some numeric coefficients satisfying four linear relations (because there are two tetrahedra not containing a given edge~$b$; denote them $t_1$ and~$t_2$, then coefficients at $\partial_{t_1}$, $\partial_{t_2}$, $x_{t_1}$ and~$x_{t_2}$ in~\eqref{apx} must vanish). For a generic~$F$, this yields a one-dimensional space of~$\alpha$'s (see, however, the Remark below).

Explicitly, for the edge $b=12$, one gets
\begin{multline*}
\alpha_{2345} = \alpha_{1345} = 0,\qquad \alpha_{1245}=\varphi_{134}\varphi_{235}-\varphi_{135}\varphi_{234},\\
 \alpha_{1235}=\varphi_{134}\varphi_{245}-\varphi_{145}\varphi_{234},\qquad \alpha_{1234}=\varphi_{135}\varphi_{245}-\varphi_{145}\varphi_{235}
\end{multline*}
as one of the solutions, so any edge operator is, in the case of general~$F$, proportional to
\begin{multline}\label{phiii}
(\varphi_{134}\varphi_{235}-\varphi_{135}\varphi_{234}) \partial_{1245} 
+(\varphi_{134}\varphi_{245}-\varphi_{145}\varphi_{234}) \partial_{1235} \\
+(\varphi_{135}\varphi_{245}-\varphi_{145}\varphi_{235}) \partial_{1234} \\
-(\varphi_{124}\varphi_{135}\varphi_{245}-\varphi_{125}\varphi_{134}\varphi_{245}-\varphi_{124}\varphi_{145}\varphi_{235}+\varphi_{125}\varphi_{145}\varphi_{234}) x_{1245} \\
+(\varphi_{123}\varphi_{135}\varphi_{245}-\varphi_{123}\varphi_{145}\varphi_{235}-\varphi_{125}\varphi_{134}\varphi_{235}+\varphi_{125}\varphi_{135}\varphi_{234}) x_{1235} \\
-(\varphi_{123}\varphi_{134}\varphi_{245}-\varphi_{124}\varphi_{134}\varphi_{235}-\varphi_{123}\varphi_{145}\varphi_{234}+\varphi_{124}\varphi_{135}\varphi_{234}) x_{1234}\,.
\end{multline}
\end{xmp}

\begin{remark}
Of course, there are degenerate cases yielding the space of~$\alpha$'s of more than one dimension. In particular, if $F=0$, then this space is three-dimensional. In this paper we are, however, most interested in the general case.
\end{remark}

\begin{remark}
An edge operator for any other edge $b=ij$ can be obtained from~\eqref{phiii} by doing any permutation of indices such that $1\mapsto i$ and~$2\mapsto j$, if we also assume that $\varphi_{klm}$ is \emph{totally antisymmetric} in its indices (while $\partial_{pqrs}$ is totally \emph{symmetric}).
\end{remark}

\begin{lemma}\label{l:w}
In the general position case, the 10 edge operators span the whole maximal isotropic subspace~$V$.
\end{lemma}

\begin{proof}
Direct calculation.
\end{proof}

\subsection{Normalization of edge operators and the 2-cocycle}\label{ss:norm}

We want now to \emph{normalize} the edge operators~$d_b$ in some canonical way. Remarkably, such normalization exists and, moreover, reveals some beautiful properties of edge operators. It turns out that we must \emph{orient} all edges in our 4-simplex~$12345$: their (default) orientation is given, by definition, by the increasing order of vertices: $b=ij$, \ $i<j$, and if we want to change the orientation of an edge, this will imply changing the sign of the edge operator:
\begin{equation}\label{i-j}
d_{ji}=-d_{ij}.
\end{equation}
We also need some standard cohomological notions: a $\mathbb C$-valued function of vertices/edges/2-faces is called 0-cochain/1-cochain/2-cochain; there is also the coboundary operator~$\delta$ with coboundaries in its image and cocycles in its kernel. The operator~$d_c$ corresponding to a 1-cochain~$c$ is, by definition, the corresponding linear combination of edge operators (if their normalization has been chosen), namely,
\begin{equation}\label{dc}
d_c = \sum_{\text{all edges }b} c(b)d_b.
\end{equation}

The four edge operators corresponding to edges having a common vertex~$i$ are linearly dependent. This follows from the fact that they belong to the three-dimensional subspace of~$V$ not containing $\partial_t$ and~$x_t$, where tetrahedron~$t$ lies opposite vertex~$i$. Remarkably, these linear dependencies can be described as follows.

\begin{theorem}\label{th:1b}
The edge operators for the ten oriented edges can be normalized in such way that zero will correspond to any 1-coboundary. Such normalization is unique up to a common factor.
\end{theorem}

In other words, with this normalization, the operator~$d_{\delta i}$, corresponding to the coboundary~$\delta i$ of any vertex~$i$ according to~\eqref{dc}, vanishes.

\begin{proof}
First, we choose \emph{some} edge operators~$d_b$, for instance, according to formula~\eqref{phiii} and the second Remark after it. For convenience, we also choose their signs in such way that the antisymmetry condition~\eqref{i-j} holds.

Second, we calculate the linear dependencies between them, mentioned in the paragraph before this Theorem, namely, for each vertex~$i$, the coefficients~$\beta_{ij}$ in
\begin{equation}\label{beta}
\sum_{\substack{j\neq i\\ i\text{ fixed}}} \beta_{ij}d_{ij}=0,\qquad i=1,\dots,5.
\end{equation}
Explicit expressions for~$\beta_{ij}$ are cumbersome and are not written out here; calculations with them have been done using computer algebra.

Third, as $\beta_{ij}$'s in each of five linear dependencies in~\eqref{beta} are determined only to within a multiplicative constant (depending only on~$i$), we consider their \emph{ratios}~$\beta_{ij}/\beta_{ik}$, and unexpectedly find that
\begin{equation}\label{ijk}
\frac{\beta_{ij}}{\beta_{ik}} \frac{\beta_{jk}}{\beta_{ji}} \frac{\beta_{ki}}{\beta_{kj}} = 1
\end{equation}
for any triangle~$ijk$.

Fourth, \eqref{ijk} means that the values $\gamma_{ij}=\beta_{ij}/\beta_{ji}$ make a \emph{multiplicative cocycle}:
\[
\gamma_{ij}\gamma_{jk}\gamma_{kj}=1,\qquad \gamma_{ji}=\gamma_{ij}^{-1},
\]
which is also --- as everything happens in a single 4-simplex --- a multiplicative coboundary:
\[
\gamma_{ij}=\frac{\mathrm c_i}{\mathrm c_j},
\]
with some nonzero values~$\mathrm c_i$ in the vertices.

Fifth, consider the normalized edge operators
\begin{equation}\label{norm}
d_{ij}^{\mathrm{norm}} = \mathrm c_i \beta_{ij}d_{ij},
\end{equation}
then
\[
d_{ji}^{\mathrm{norm}}=-d_{ij}^{\mathrm{norm}}
\]
and
\[
\sum_{\substack{j\neq i\\ i\text{ fixed}}}d_{ij}^{\mathrm{norm}}=0,\qquad i=1,\dots,5.
\]
\end{proof}

From now on, we fix a (nonzero) normalization according to Theorem~\ref{th:1b}, and write simply~$d_b$ instead of~$d_b^{\mathrm{norm}}$.

\medskip

There are 10 edges in our 4-simplex~$12345$, and only 4 linearly independent 1-coboundaries. As the space spanned by the edge operators, for a given~$\mathcal W$, must be 5-dimensional, there must be one more linear dependence between them --- one more (non-cobound\-ary) 1-cochain~$\nu$ whose corresponding operator vanishes:
\begin{equation}\label{nu}
d_{\nu}=0.
\end{equation}
As neither adding any coboundary to~$\nu$ nor multiplying~$\nu$ by a nonzero factor changes the set of linear dependencies between edge operators, the essential part of~$\nu$ is contained in \emph{its} coboundary --- 2-cocycle
\begin{equation}\label{odn}
\omega=\delta\nu,
\end{equation}
taken also to within a nonzero factor.

\begin{dfn}\label{dfn:W-cocycle}
A nonvanishing 2-cocycle~$\omega$ corresponding to a quasi-Gaussian 4-simplex weight~$\mathcal W$ as explained above is called \emph{$\mathcal W$-cocycle}.
\end{dfn}

\section{2-cocycles parameterize 4-simplex weights}\label{s:pw}

\subsection{The general theorem}\label{ss:pw}

\begin{theorem}\label{th:2c}
Nonvanishing 2-cocycles~$\omega$ on a 4-simplex parameterize quasi-Gaussian weights~$\mathcal W$ in the following sense:
\begin{enumerate}\itemsep 0pt
 \item\label{i:v} generic~$\mathcal W$ determines $\omega$ --- its $\mathcal W$-cocycle, see Definition~\ref{dfn:W-cocycle} --- up to a factor,
 \item\label{i:w} generically, five (${}={}$maximal number of) independent ratios of the values~$\omega_{ijk}$ --- components of our 2-cycle~$\omega$ --- determine the maximal isotropic space for~$\mathcal W$ to within renormalizations
 \begin{equation}\label{renorm}
 x_t \mapsto \lambda_t x_t, \quad \partial_t \mapsto (1/\lambda_t) \partial_t, \quad \lambda_t\in \mathbb C,
 \end{equation}
 of five boundary Grassmann variables and possible interchanges
 \begin{equation}\label{interc}
 \partial_t \leftrightarrow x_t.
 \end{equation}
\end{enumerate}
\end{theorem}

\begin{dfn}\label{dfn:gauge}
We call the transformations of quasi-Gaussian weights corresponding to \eqref{renorm} and~\eqref{interc} \emph{gauge transformations}, and refer to the possibility of making gauge transformations as \emph{gauge freedom}.
\end{dfn}

\begin{remark}
Gauge transformations \eqref{renorm} and~\eqref{interc} are, of course, simply orthogonal transformations in the subspace~$\mathcal V_t$ spanned by $\partial_t$ and~$x_t$, and any such orthogonal transformation is either~\eqref{renorm} or its composition with~\eqref{interc}.
\end{remark}

\begin{proof}[Proof of Theorem~\ref{th:2c}]
\begin{itemize}\itemsep 0pt
 \item[\ref{i:v}] This has been already explained in Subsection~\ref{ss:norm}.
 \item[\ref{i:w}]
First, some~$\mathcal W$ does correspond to any point in a Zariski open subset in the space of parameters~$w_m$, where $w_m$, $m=1,\ldots,5$, are the independent ratios of the values~$\omega_{ijk}$ mentioned in the Theorem. To show this, it is enough to consider the ``even'' case and note that the rank of the ($5\times 10$) Jacobian matrix $(\partial w_m / \partial \varphi_{ijk})$ is generically~5, which can be checked by a direct calculation.

Now, let there be two 4-simplex quasi-Gaussian weights $\mathcal W$ and~$\mathcal W'$ with the two respective sets $\{d_a\}$ and~$\{d'_a\}$ of normalized edge operators, and let the above mentioned parameters~$w_m$ be the same for $\mathcal W$ and~$\mathcal W'$. For a pair $a,b$ of edges, the scalar product between $d_a$ and~$d_b$ --- which of course vanishes --- is a sum of ``partial scalar products'' over five tetrahedra, compare formulas \eqref{sc} and~\eqref{sct}:
\begin{equation}\label{psc}
\langle d_a, d_b \rangle = \sum_{t\subset 12345} \langle d_a, d_b \rangle_t.
\end{equation}
We will show that these partial scalar products within the two sets are the same up to a factor:
\begin{equation}\label{dt}
\langle d_a, d_b \rangle_t = c \langle d'_a, d'_b \rangle_t
\end{equation}
with the same~$c$ for all pairs $a,b$ of edges and all~$t$. Once this is done, it is an easy exercise to show (see the Remark before this proof) that every~$d_a$ is obtained from~$d'_a$, first, by multiplying it by~$\sqrt{c}$, and then, by a set (independent of~$a$) of renormalizations~\eqref{renorm} and --- maybe for some~$t$ --- interchanges~\eqref{interc}. Due to Lemma~\ref{l:w}, this will be enough to prove our Theorem. So, below we prove formula~\eqref{dt}.

There are three pairs of opposite edges in tetrahedron~$t$; denote the first pair as~$a_1,b_1$, the second~$a_2,b_2$, and the third~$a_3,b_3$. As the intersection of stars of two opposite edges is just one tetrahedron~$t$, it follows that
\begin{equation}\label{ab}
\langle d_{a_1}, d_{b_1} \rangle_t = \langle d_{a_1}, d_{b_1} \rangle = 0,
\end{equation}
and similarly
\begin{equation}\label{abp}
\langle d_{a_2}, d_{b_2} \rangle_t = 0, \quad \langle d_{a_3}, d_{b_3} \rangle_t = 0.
\end{equation}

We denote the subspace of~$\mathcal V$ spanned by $\partial_t$ and~$x_t$ as~$\mathcal V_t$ (like we did in Subsection~\ref{ss:C}. To the direct sum decomposition $\mathcal V=\bigoplus_{t\subset 12345}\mathcal V_t$ corresponds the decomposition of operators for which we will use the notations like $d_a=\sum_{t\subset 12345}d_a|_t$, and $d_a|_t$ will be called \emph{$t$-component} of~$d_a$. Clearly, replacing operators with their $t$-components does not change partial scalar products associated with tetrahedron~$t$.

There are four linear dependencies between $d_{a_1}$, $d_{b_1}$, $d_{a_2}$, $d_{b_2}$, $d_{a_3}$ and~$d_{b_3}$, and hence between their $t$-components: three of them arise from Theorem~\ref{th:1b} (and thus have fixed coefficients~$\pm 1$), and one more comes from the 2-cocycle~$\omega$. A small exercise in linear algebra (actually done in the proof of Theorem~\ref{th:1:iso} below) shows that these linear dependencies, together with \eqref{ab} and~\eqref{abp}, fix the partial scalar product $\langle d_{\ldots}, d_{\ldots} \rangle_t$ up to a factor. Moreover, the special form of our linear dependencies is responsible for the fact that such a nontrivial scalar product exists at all.

As the analogues of \eqref{ab} and~\eqref{abp}, as well as the same four linear dependencies, apply also to primed operators~$d'_{\ldots}$, we have proved~\eqref{dt} \emph{for each tetrahedron~$t$ separately}; what remains is to prove that the factor~$c$ does not depend on~$t$.

It is enough to show how the ratio of these factors is fixed for two \emph{adjacent} tetrahedra. Call them $t$ and~$t'$, and their common 2-face~$ijk$. We consider operators $d_{ij}$ and~$d_{ik}$; their scalar product involves only tetrahedra $t$ and~$t'$ and must vanish:
\[
\langle d_{ij}, d_{ik} \rangle = \langle d_{ij}, d_{ik} \rangle_t + \langle d_{ij}, d_{ik} \rangle_{t'} = 0.
\]
This clearly fixes the mentioned ratio.
\end{itemize}
\end{proof}

\subsection{Superisotropic operators: definition and construction from a given 2-cocycle}\label{ss:s-iso}

Theorem~\ref{th:2c} gives no explicit expression for a quasi-Gaussian weight~$\mathcal W$ in terms of a given 2-cocycle~$\omega$ --- its item~\ref{i:w} just states that the components of~$\mathcal W$ are determined by a system of algebraic equations, up to transformations corresponding to \eqref{renorm} and~\eqref{interc}. Using, if needed, properly chosen interchanges~\eqref{interc}, we can assume that $\mathcal W$ is a Gaussian weight~\eqref{gg}. Restricting ourself to this case, we are going to show how a weight~\eqref{gg} can be retrieved efficiently from a generic cocycle~$\omega$.

Here we define our main tool to be used for this purpose --- \emph{superisotropic operators}, and give their explicit construction starting from a given $\mathcal W$-cocycle~$\omega$. Then, we study further our operators and their relations with matrix~$F$~\eqref{F} in Subsection~\ref{ss:compare}, and give some elegant explicit formulas using elliptic functions in Subsection~\ref{ss:e}.

\begin{dfn}\label{dfn:s-iso}
A \emph{superisotropic operator} is such an isotropic operator of the form~\eqref{bgt} that, for each 3-face~$t$ of the 4-simplex~$12345$, either $\beta_t=0$ or $\gamma_t=0$.
\end{dfn}

\begin{xmp}\label{xmp:s-iso}
The operators entering the 5-col\-umn~\eqref{pFx} are superisotropic.
\end{xmp}

An obvious characteristic property of superisotropic operators among all operators~$d$~\eqref{bgt} is that \emph{each $t$-component}
\begin{equation}\label{tc}
d|_t = \beta_t \partial_t + \gamma_t x_t
\end{equation}
is isotropic.

\begin{remark}
The definition~\eqref{tc} of $t$-component agrees, of course, with what we have already used in the proof of Theorem~\ref{th:2c}.
\end{remark}

\begin{remark}
In the notations of this paper, the superisotropicity of~$f$ can be written as either $\langle f|_t,f|_t \rangle=0$ or, equivalently, $\langle f,f \rangle_t = 0$, which must hold for all 3-faces $t\subset \{12345\}$.
\end{remark}

We will construct some superisotropic operators~$f$, including operators proportional to those entering the 5-col\-umn~\eqref{pFx}, in the form of linear combinations
\begin{equation}\label{1:f}
f = \sum_{1\le i<j\le 5} \alpha_{ij} d_{ij},\qquad \alpha_{ij}\in\mathbb C,
\end{equation}
of edge operators.

For the first operator we are going to construct, the coefficients~$\alpha_{ij}$ in~\eqref{1:f} are defined as follows. For our 2-cocycle~$\omega$, consider the square roots $\sqrt{ \omega_s }$ for all 2-faces~$s$, writing these latter as
\[
s=ijk,\qquad i<j<k.
\]
For each of the square roots, we choose arbitrarily (and fix) one of its two possible values. Similarly, we denote the edges as $b=ij$, with $i<j$. We put then
\begin{equation}\label{1:alpha}
\alpha_b = \prod_{\substack{s\supset b\\ \mathrm{or}\;s\cap b=\emptyset}} \sqrt{ \omega_s }\,.
\end{equation}
For instance, the coefficient at~$d_{12}$ is
\[
\alpha_{12} = \sqrt{ \omega_{123} } \sqrt{ \omega_{124} } \sqrt{ \omega_{125} } \sqrt{ \omega_{345} }
\]

\begin{theorem}\label{th:1:iso}
If the coefficients~$\alpha_b$ in a linear combination~\eqref{1:f} of edge operators are as in~\eqref{1:alpha}, then $f$ is superisotropic.
\end{theorem}

First, the following easy lemma. In it, we consider a tetrahedron~$t$ in our 4-simplex~$12345$ and the space~$\mathcal V_t$ spanned by two operators~$\partial_t$ and~$x_t$. Also, the six $t$-components~$d_b|_t$ of edge operators for $b\subset t$ belong to~$\mathcal V_t$ as well. Hence, there must be at least four linear relations between these~$d_b|_t$. Lemma~\ref{l:easy} precises this argument.

\begin{lemma}\label{l:easy}
For a tetrahedron~$t=ijkl\subset 12345$ and a generic 2-cocycle~$\omega$, there are exactly four independent linear relations between the $t$-components of edge operators corresponding to edges $b\subset t$. They can be taken as follows: any three of the restrictions onto~$t$ of the (four) relations $d_{\delta i}=0$, \dots, $d_{\delta l}=0$ corresponding to the coboundaries, and the restriction of~\eqref{nu}.
\end{lemma}

These relations, for $t=1234$, are written explicitly below as formulas~\eqref{5u}.

\begin{proof}
It is clear that the $t$-components indeed obey the mentioned relations; it may only be worth mentioning that edges lying outside tetrahedron~$t$ make, indeed, no contribution in these relations, because the corresponding $t$-components vanish. That there are \emph{no more} independent linear relations, follows from the easily checked fact that any two operators $d_a$ and~$d_b$, \ $a,b\subset t$, have, in a general position, linearly independent $t$-components.
\end{proof}

\begin{proof}[Proof of Theorem~\ref{th:1:iso}]
Consider, for instance, the tetrahedron $t=1234$. For it, $f|_t$ can be represented as the sum of the following three expressions:
\begin{gather}
(\alpha_{12}d_{12}+\alpha_{34}d_{34})|_{1234}, \label{pha} \\
 (\alpha_{13}d_{13}+\alpha_{24}d_{24})|_{1234} \label{phb} \\
 \text{ and}\quad (\alpha_{14}d_{14}+\alpha_{23}d_{23})|_{1234}. \label{phc}
\end{gather}
We claim that \emph{each} of operators \eqref{pha}, \eqref{phb} and~\eqref{phc} is already isotropic and, moreover, they all are \emph{proportional} to each other.

First, we note that, because the stars of edges $12$ and~$34$ have exactly one tetrahedron~$1234$ in common,
\begin{equation}\label{d}
\langle d_{12},d_{34} \rangle_{1234} = \langle d_{12},d_{34} \rangle = 0.
\end{equation}
So, $d_{12}$ and~$d_{34}$ make an orthogonal basis in the linear space~$\mathcal V_{1234}$ spanned by $\partial_{1234}$ and~$x_{1234}$, with respect to the scalar product $\langle \,\cdot\,,\,\cdot\, \rangle_{1234}$. Similar statements hold, of course, also for the two other pairs of opposite edges in~$1234$. 

Second, to find some information about the \emph{norms} of $d_{12}$ and~$d_{34}$ --- namely, the ratio $\langle d_{12},d_{12} \rangle_{1234} \,/\, \langle d_{34},d_{34} \rangle_{1234}$ --- we use the linear relations between the $t$-components of edge operators in \eqref{pha}--\eqref{phc}, which are, according to Lemma~\ref{l:easy}, as follows:
\begin{equation}\label{5u}
\left. \begin{array}{rcl}
d_{12}|_t+d_{13}|_t+d_{14}|_t & = & 0,\\
-d_{12}|_t+d_{23}|_t+d_{24}|_t & = & 0,\\
-d_{13}|_t-d_{23}|_t+d_{34}|_t & = & 0,\\
-d_{14}|_t-d_{24}|_t-d_{34}|_t & = & 0,\\
\nu_{12}d_{12}|_t+\nu_{13}d_{13}|_t+\nu_{14}d_{14}|_t +\nu_{23}d_{23}|_t+\nu_{24}d_{24}|_t+\nu_{34}d_{34}|_t & = & 0.
\end{array} \right\}
\end{equation}
Here (and below), of course, $t=1234$, and $\nu$ is a preimage of~$\omega$, see~\eqref{odn}.

From the system above, we deduce that
\begin{equation}\label{dd}
d_{13}|_t=-\frac{\omega_{124}d_{12}|_t+\omega_{234}d_{34}|_t}{\omega_{134}-\omega_{234}}\,,\qquad
d_{24}|_t=-\frac{\omega_{123}d_{12}|_t+\omega_{134}d_{34}|_t}{\omega_{134}-\omega_{234}}\,.
\end{equation}
Now the condition $\langle d_{13},d_{24} \rangle_t = 0$ (similar to~\eqref{d}) leads, together with \eqref{d} and~\eqref{dd}, to the following relation:
\begin{equation}\label{nd}
\omega_{123}\omega_{124}\langle d_{12},d_{12} \rangle_{1234} + \omega_{134}\omega_{234}\langle d_{34},d_{34} \rangle_{1234} = 0,
\end{equation}
and it follows directly from \eqref{nd}, \eqref{d} and~\eqref{1:alpha} that the operator~\eqref{pha} is isotropic.

The proportionality between operators \eqref{pha} and~\eqref{phb} follows, of course, from the explicit expressions~\eqref{dd}. Similarly, we can prove that the operator~\eqref{phc} is proportional to both of them.
\end{proof}

We are constructing the isotropic space~$V$ of operators annihilating a Gaussian exponent \eqref{gg},~\eqref{Phi}. There exists a (nonzero) superisotropic operator $f\in V$ whose all components are proportional to differentiations (because matrix~$F$~\eqref{F} is of rank at most~4, so there exists a nontrivial but vanishing linear combination of its rows. Our operator~$f$ is then the linear combination of \emph{differentiations} with the same coefficients):
\[
f|_t \propto \partial_t .
\]
We identify, by definition, this operator with our operator~$f$ defined by \eqref{1:f} and~\eqref{1:alpha}.

Recall that we were using one fixed choice of square root signs. If we change some of these signs, some of the $t$-components of the new operator will no longer be proportional to their old versions. As such a component is still isotropic, it must be proportional to the operator~$x_t$.

This way we get superisotropic operators proportional to the entries of column~\eqref{pFx}. Namely, we change \emph{two} properly chosen square root signs and obtain such~$f^{(t)}$ that
\[
f^{(t)}|_t \propto \partial_t \quad \text{for one }t \text{ and}\quad f^{(t)}|_{t'} \propto x_{t'} \quad \text{for the four }t'\ne t .
\]
For instance, if $t=2345$, then the roots whose signs are to be changed can be chosen in any of the following three ways: as $\sqrt{\omega_{123}}$ and~$\sqrt{\omega_{145}}$, or as $\sqrt{\omega_{124}}$ and~$\sqrt{\omega_{135}}$, or as $\sqrt{\omega_{125}}$ and~$\sqrt{\omega_{134}}$. Each of these ways leads to the following:
\begin{itemize}\itemsep 0pt
 \item $\alpha_{12}$, $\alpha_{13}$, $\alpha_{14}$ and~$\alpha_{15}$ change their signs,
 \item the rest of~$\alpha_{ij}$ do not change.
\end{itemize}
For the other four tetrahedra~$t$, square roots are chosen in obvious analogy with the above.

\subsection{The ratio of $t$-components in two superisotropic operators, and double ratios of matrix~$F$ entries}\label{ss:compare}

We will obtain what can be called \emph{double ratios} of matrix entries in~\eqref{F}. These entries must be nonvanishing and lie in two rows and two columns, and our double ratio is by definition the product of two of them lying on one diagonal divided by the product of two lying on the other diagonal, see Example~\ref{xmp:dr} below. Such double ratios are invariant under gauge transformations~\eqref{renorm}, and can be easily seen to determine all the entries in~\eqref{F} to within these transformations, so they are almost all we can obtain from a cocycle~$\omega$ according to Theorem~\ref{th:2c}, item~\ref{i:w}. Here `almost' means `to within an arbitrariness of discrete character' which is related to possible interchanges~\eqref{interc}: an \emph{even} number of such interchanges may lead from one Gaussian exponent to another.

\begin{remark}
What happens with matrix~$F$~\eqref{F} under a transformation~\eqref{renorm} is of course
\begin{equation}\label{AFA}
F\mapsto AFA,\qquad A=\diag(\{\lambda_t\}_{t\subset\{12345\}}).
\end{equation}
\end{remark}

As our operators~$f^{(t)}$ are proportional to the entries of column~\eqref{pFx}, the double ratios of entries in~\eqref{F} are equal to the similar double ratios of components of~$f^{(t)}$ (which makes sense as the components are proportional).

\begin{xmp}\label{xmp:dr}
In the equality~\eqref{ffff}, the l.h.s.\ is a double ratio for the first two rows and last two columns of matrix~$F$, while the r.h.s.\ is the corresponding double ratio of components of~$f^{(t)}$.
\begin{equation}\label{ffff}
\frac{\varphi_{235}\varphi_{134}}{\varphi_{135}\varphi_{234}} = \frac{ f^{(2345)}|_{1235}\, f^{(1345)}|_{1234} }{ f^{(1345)}|_{1235}\, f^{(2345)}|_{1234} }.
\end{equation}
\end{xmp}

We now explain how to calculate the double ratios such as one in the r.h.s.\ of~\eqref{ffff} in terms of~$\omega$. It is the product of two ratios: $f^{(2345)}|_{1235}/f^{(1345)}|_{1235}$ and~$f^{(1345)}|_{1234}/f^{(2345)}|_{1234}$. As they are similar, we consider below only the second of them, which we denote~$\varkappa$, that is,
\begin{equation}\label{tf}
\tilde f|_{1234}=0
\end{equation}
for
\begin{equation}\label{tf1}
\tilde f = f^{(1345)}-\varkappa f^{(2345)}.
\end{equation}

\begin{remark}
Recall that elegant formulas in terms of elliptic functions have been promised to appear in Subsection~\ref{ss:e}. They will embrace all entries of the matrix~$F$~\eqref{F}. Right now we see that all such quantities as~\eqref{ffff} can be calculated if we can calculate~$\varkappa$ and a few similar quantities. So, below in this Subsection we just explain how to calculate~$\varkappa$, and write out the explicit expression \eqref{vk1},~\eqref{vk2} for it. A comparison of this explicit expression with formula~\eqref{e} below shows that introducing elliptic functions indeed makes sense.
\end{remark}

To calculate~$\varkappa$, we note that the coefficients~$\tilde \alpha_{ij}$ in the decomposition~\eqref{1:f} of~$\tilde f$:
\begin{equation*}
\tilde f = \sum_{1\le i<j\le 5} \tilde \alpha_{ij} d_{ij},
\end{equation*}
determine, due to~\eqref{tf}, a linear relation that must follow from the relations indicated in Lemma~\ref{l:easy}. That is,
\[
f|_{1234}=0\;\; \Leftrightarrow \;\;\rank \begin{pmatrix}
\tilde \alpha_{12}&\tilde \alpha_{13}&\tilde \alpha_{14}&\tilde \alpha_{23}&\tilde \alpha_{24}&\tilde \alpha_{34} \\
1&1&1&0&0&0 \\
-1&0&0&1&1&0 \\
0&-1&0&-1&0&1 \\
0&0&-1&0&-1&-1 \\
\nu_{12}& \nu_{13}& \nu_{14}& \nu_{23}& \nu_{24}& \nu_{34} \\
\end{pmatrix} = 4 \,
\]
(compare with the system~\eqref{5u}). So, we calculate the~$\tilde \alpha_{ij}$'s according to~\eqref{tf1}, that is, $\tilde \alpha_{ij}=\alpha_{ij}^{(1345)}-\varkappa \alpha_{ij}^{(2345)}$, where we substitute the expressions~\eqref{1:alpha} with the square roots signs changed relevantly (see the end of Subsection~\ref{ss:s-iso}), and obtain the following answer:
\begin{equation}\label{vk1}
\varkappa=\frac{\lambda_+}{\lambda_-}
\end{equation}
where
\begin{multline}\label{vk2}
\lambda_{\pm} = \omega_{124}\sqrt{\omega_{125}} \sqrt{\omega_{345}} - \omega_{123}\sqrt{\omega_{125}} \sqrt{\omega_{345}} \pm \bigl( -\sqrt{\omega_{123}}\sqrt{\omega_{135}} \sqrt{\omega_{234}}\sqrt{\omega_{245}} \\
 + \sqrt{\omega_{124}} \sqrt{\omega_{134}}\sqrt{\omega_{135}}\sqrt{\omega_{245}} + \sqrt{\omega_{124}}\sqrt{\omega_{145}}\sqrt{\omega_{234}} \sqrt{\omega_{235}} \\
 - \sqrt{\omega_{123}}\sqrt{\omega_{134}} \sqrt{\omega_{145}}\sqrt{\omega_{235}}\, \bigr). 
\end{multline}

\subsection{Explicit formulas in elliptic functions}\label{ss:e}

Elegant formulas appear if we parameterize our 2-cocycle~$\omega$ using Jacobi elliptic functions. We remind that we are still studying objects associated with a single 4-simplex~$12345$. We will turn to the whole Pachner move 3--3 below in Section~\ref{s:33}, and in that case we cannot yet present such simple and explicit formulas. This Subsection is, however, written in the hope that its ideas will be eventually of use for studying more global objects (than a single 4-simplex) as well. Apparently, more general functions of algebraic-geometrical origin will be needed.

\begin{lemma}\label{l:1:c}
Let there be a triangulated manifold~$M$ with a complex number~$x_i$ attached to each vertex~$i$. Let\/ $\sn$ be the Jacobi elliptic sine with a fixed modulus~$\kappa$:
\[
\sn x = \sn(x,\kappa).
\]
Then the cochain
\begin{equation}\label{1:c}
\omega_{ijk}= \sn(x_i-x_j)\sn(x_i-x_k)\sn(x_j-x_k),
\end{equation}
is exact:
\[
\omega = \delta \nu
\]
for some 1-cochain~$\nu$.
\end{lemma}

\begin{proof}
Define
\[
\nu_{ij}=\frac{1}{\kappa^2} \frac{\sn(x_i-x_j)}{\sn x_i \sn x_j} .
\]
A small exercise in elliptic functions shows then that, indeed,
\[
\omega_{ijk} = \nu_{jk}-\nu_{ik}+\nu_{ij}.
\]
\end{proof}

\begin{theorem}\label{th:1:c}
A generic 2-cocycle~$\omega$ on our 4-simplex~$12345$ can be represented, up to a common factor, using the elliptic parameterization~\eqref{1:c}.
\end{theorem}

\begin{proof}
As Lemma~\ref{l:1:c} affirms that \eqref{1:c} is indeed a cocycle, it remains to note the following. The linear space of 2-cocycles on a 4-simplex is six-dimensional; taken up to a factor, they form a five-dimensional \emph{projective} space. On the other hand, the modulus~$\kappa$ and four differences between the~$x_i$'s make together also a five-dimensional parameter space, and a calculation shows that the Jacobian matrix of the mapping from the latter space to the former (made, of course, according to~\eqref{1:c}) has, generically, the full rank.
\end{proof}

\begin{remark}
Explicit formulas for $\kappa$ and differences between the~$x_i$'s in terms of~$\omega$ can also be obtained. We do not write them out here. Such formulas deserve a separate study aimed at revealing their algebraic-geometrical nature and, as we have already said, in a wider setting than just one 4-simplex.
\end{remark}

An exercise in elliptic functions shows that, given~\eqref{1:c}, the quantity~$\varkappa$, introduced in Subsection~\ref{ss:compare}, is
\begin{multline*}
\varkappa = \frac{f^{(1345)}|_{1234}}{f^{(2345)}|_{1234}} = - \frac{\sn}{\cn\dn}\left(\frac{x_1-x_3}{2}\right) \frac{\cn\dn}{\sn}\left(\frac{x_2-x_3}{2}\right)\\
 \cdot\frac{\sn}{\cn\dn}\left(\frac{x_1-x_4}{2}\right) \frac{\cn\dn}{\sn}\left(\frac{x_2-x_4}{2}\right).
\end{multline*}
Calculating also all other ratios like~$\varkappa$, we come to the following theorem.

\begin{theorem}\label{th:F-ell}
For the cocycle~$\omega$ given by~\eqref{1:c}, a corresponding 4-simplex Grassmann weight is given by \eqref{gg}, \eqref{Phi}, with the following quantities~$\varphi_{ijk}$, thought of as entries~$F_{ij}$ of matrix~\eqref{F}:
\begin{equation}\label{e}
F_{ij}=\frac{\sn}{\cn\dn}\left(\frac{x_i-x_j}{2}\right).
\end{equation}
\qed
\end{theorem}

Recall that other weights corresponding to the same cocycle~$\omega$ can be obtained using the continuous transformations~\eqref{AFA}, corresponding to~\eqref{renorm}, and discrete transformations corresponding to~\eqref{interc} (explicit formulas for the latter transformations can also be written out, but we do not give them here).

\section{Edge operators and the move 3--3}\label{s:33}

\subsection{Pachner move 3--3}\label{ss:m33}

Here we describe a move 3--3 and fix notations for the involved vertices and simplices. Let there be a cluster of three 4-simplices 12345, 12346 and~12356 situated around the 2-face~123. Pachner move 3--3 transforms it into the cluster of three other 4-simplices, 12456, 13456 and~23456, situated around the 2-face~456. The inner 3-faces (tetrahedra) are 1234, 1235 and~1236 in the l.h.s., and 1456, 2456 and~3456 in the r.h.s. The boundary of both sides consists of nine tetrahedra; we like to arrange them in the following table, where also 4-simplices are indicated by small numbers to which the tetrahedra belong:
\begin{gather}\label{table}
\begin{array}{c|ccc}
 & \scriptstyle 12456 & \scriptstyle 13456 & \scriptstyle 23456 \\ \hline
\scriptstyle 12345 & 1245 & 1345 & 2345 \\
\scriptstyle 12346 & 1246 & 1346 & 2346 \\
\scriptstyle 12356 & 1256 & 1356 & 2356
\end{array}
\end{gather}
Thus, the tetrahedra in every row correspond to a 4-simplex in the l.h.s., and the tetrahedra in every column correspond to a 4-simplex in the r.h.s.\ of the move.

\subsection{The 3--3 relation in Grassmann algebra and the annihilating operator spaces for its sides}\label{ss:r33}

The general form of the 3--3 relation considered in this paper uses the notations for vertices and simplices adopted in Subsection~\ref{ss:m33} and is
\begin{multline}\label{33}
\iiint \mathcal W_{12345} \mathcal W_{12346} \mathcal W_{12356} \,\mathrm dx_{1234} \,\mathrm dx_{1235} \,\mathrm dx_{1236} \\
 = \const \iiint \mathcal W_{12456} \mathcal W_{13456} \mathcal W_{23456} \,\mathrm dx_{1456} \,\mathrm dx_{2456} \,\mathrm dx_{3456},
\end{multline}
where each $\mathcal W_{ijklm}$ is a Grassmann algebra element containing only those Grassmann variables~$x_t$ that belong to the 3-faces of 4-simplex~$ijklm$. More specifically, we take as~$\mathcal W_{ijklm}$ quasi-Gaussian weights defined according to Definition~\ref{dfn:W}. For a moment, we will consider the l.h.s.\ and r.h.s.\ of~\eqref{33} separately, just as expressions made of such 4-simplex weights, and putting aside their equalness.

\begin{theorem}\label{th:h33}
The l.h.s.\ and r.h.s.\ of~\eqref{33} are annihilated by nine-dimensional (i.e., maximal) isotropic spaces of operators of the form
\begin{equation}\label{dlhs}
d = \sum_{\mathrm{boundary}\;t} (\beta_t\partial_t+\gamma_tx_t),
\end{equation}
where the sum is taken over the tetrahedra in table~\eqref{table}.
\end{theorem}

These nine-dimensional spaces will be the linear spans of some specific operators that we construct explicitly, for the l.h.s.\ of~\eqref{33}, in the following Lemma~\ref{l:lhs33}.

\begin{lemma}\label{l:lhs33}
Let
\begin{gather*}
d'=\sum_{t\subset 12345} (\beta'_t\partial_t+\gamma'_tx_t),\quad d''=\sum_{t\subset 12346} (\beta''_t\partial_t+\gamma''_tx_t), \\ 
\text{and} \quad d'''=\sum_{t\subset 12356} (\beta'''_t\partial_t+\gamma'''_tx_t)
\end{gather*}
be some operators annihilating the respective 4-simplex weights in the l.h.s.\ of~\eqref{33}:
\begin{equation}\label{ddd}
d'\mathcal W_{12345}=0,\quad d''\mathcal W_{12346}=0,\quad d'''\mathcal W_{12356}.
\end{equation} 
Let, moreover, $d'$, $d''$ and~$d'''$ be such that their coefficients belonging to the inner tetrahedra agree in the following way:
\begin{equation}\label{agree-d}
 \left.\begin{array}{rl}
  \beta'_{1234}=\beta''_{1234},&\quad \gamma'_{1234}=-\gamma''_{1234},\quad\\[.3ex]
  \beta'_{1235}=\beta'''_{1235},&\quad \gamma'_{1235}=-\gamma'''_{1235},\\[.3ex]
  \beta''_{1236}=\beta'''_{1236},&\quad \gamma''_{1236}=-\gamma'''_{1236}.
 \end{array}\right\}
\end{equation}
Then, the l.h.s.\ of~\eqref{33} --- denote it~$\mathcal W_{\mathrm{l.h.s.}}$ --- is annihilated by the operator~$d$ given by~\eqref{dlhs}, where the coefficients $\beta_t$ and~$\gamma_t$ are as follows:
\begin{eqnarray*}
\beta_t=\beta'_t,\quad \gamma_t=\gamma'_t && \text{for } \;t\; \text{ in the first row in~\eqref{table}},\\
\beta_t=\beta''_t,\quad \gamma_t=\gamma''_t && \text{for } \;t\; \text{ in the second row in~\eqref{table}},\\
\beta_t=\beta'''_t,\quad \gamma_t=\gamma'''_t && \text{for } \;t\; \text{ in the third row in~\eqref{table}}.
\end{eqnarray*}
\end{lemma}

\begin{proof}
First, we take the integrand $\mathfrak W_{\mathrm{l.h.s.}}=\mathcal W_{12345} \mathcal W_{12346} \mathcal W_{12356}$ and apply to it our operator~$d$~\eqref{dlhs}. Using~\eqref{ddd}, applying the Leibniz rule~\eqref{L} and keeping in mind the fact that each~$\mathcal W_t$ is either even or odd, we arrive at the equality
\[
d\mathfrak W_{\mathrm{l.h.s.}} = -(\beta'_{1234}\partial_{1234} + \beta'_{1235}\partial_{1235} + \beta''_{1236}\partial_{1236})\mathfrak W_{\mathrm{l.h.s.}}.
\]
This obviously gives the desired relation $d\mathcal W_{\mathrm{l.h.s.}}=0$.
\end{proof}

\begin{proof}[Proof of Theorem~\ref{th:h33}]
It remains to note that a direct calculation shows that the linear span of operators constructed in Lemma~\ref{l:lhs33} is indeed nine-dimensional, and that there exists the obvious analogue of Lemma~\ref{l:lhs33} for the r.h.s.\ of \eqref{33}.
\end{proof}

\begin{theorem}\label{th:i33}
The equality~\eqref{33} holds, with some value~$\const$, provided the annihilating subspaces of operators~\eqref{dlhs} for the l.h.s.\ and r.h.s.\ coincide.
\end{theorem}

\begin{proof}
This follows directly from Theorem~\ref{th:v}, item~\ref{i:t}.
\end{proof}

\subsection{2-cocycle on $\partial\Delta^5$ and 4-simplex weights}\label{ss:6}

Recall that $\partial\Delta^5$ --- the boundary of a 5-simplex~$\Delta^5$ --- is the union of the l.h.s.\ and r.h.s.\ of our Pachner move 3--3. Let there be a 2-cocycle~$\Omega$ on~$\Delta^5$.

Given~$\Omega$, we are going to construct 4-simplex weights satisfying~\eqref{33}. Namely:
\begin{itemize}
 \item each of the six weights~$\mathcal W_t$ corresponds to the respective restriction~$\omega=\Omega|_t$ of~$\Omega$ onto the 4-simplex~$t$ according to Subsection~\ref{ss:pw},
 \item the Grassmann variables for all~$t$ are chosen in a consistent way (details below; recall that the freedom in choosing these variables is described by formulas \eqref{renorm} and~\eqref{interc}).
\end{itemize}
Then we will see that, for a given edge~$a$, the edge operators for separate 4-simplices~$u\supset a$ --- we denote them~$d_a^{(u)}$ --- compose together well, according to the construction in Lemma~\ref{l:lhs33}. Their compositions are edge operators acting in the space~$\mathfrak V$ corresponding to the \emph{boundary} of either l.h.s.\ or r.h.s.\ of the move, consisting of the tetrahedra in table~\eqref{table} (these edge operators are defined according to the same Definition~\ref{dfn:edge}, with $\mathfrak V$ in place of~$V$). Moreover, we get thus the same edge operators for the l.h.s.\ as for the r.h.s., and it can be checked that they span a nine-dimensional linear space. So, according to Theorem~\ref{th:i33}, \eqref{33} does hold.

There are consistency conditions for the inner tetrahedra (within either l.h.s.\ or r.h.s.) and for the boundary tetrahedra (between the two sides), and they all must also be shown to be consistent between themselves.

\subsubsection*{Consistency conditions for the inner tetrahedra}

For a given inner tetrahedron~$t$ situated between 4-simplices $u_1$ and~$u_2$, we consider two partial scalar products (compare~\eqref{psc}) between the edge operators, namely
\[
\langle d_a^{(u_1)}, d_b^{(u_1)} \rangle_t\quad \text{and}\quad \langle d_a^{(u_2)}, d_b^{(u_2)} \rangle_t,\qquad a,b\subset t=u_1\cap u_2.
\]
We note that such partial scalar products are necessarily proportional:
\[
\langle d_a^{(u_1)}, d_b^{(u_1)} \rangle_t = c \langle d_a^{(u_2)}, d_b^{(u_2)} \rangle_t,\qquad c\text{ independent from }a\text{ and }b.
\]
This is proved in the same way as the proportionality~\eqref{dt} in the proof of Theorem~\ref{th:2c}, taking into account that the coefficients of linear dependencies between edge operators for edges in a tetrahedron~$t$ depend only on the restriction of our 2-cocycle onto~$t$. We can choose the common factors at the scalar products for $u_1$ and~$u_2$ in such way that $c=-1$:
\begin{equation}\label{0}
\langle d_a^{(u_1)}, d_b^{(u_1)} \rangle_t = - \langle d_a^{(u_2)}, d_b^{(u_2)} \rangle_t.
\end{equation}

Then, we choose the Grassmann variable~$x_t$ in such way that the $t$-components (recall that these have been introduced in the paragraph after formula~\eqref{abp}) of edge operators for $u_1$ and~$u_2$ have the same expressions in terms of $\partial_t$ and~$x_t$, except that the signs at~$x_t$ are opposite:
\begin{equation}\label{agree-t}
 \begin{array}{rl} \text{if} & d_a^{(u_1)}|_t = \beta_t \partial_t + \gamma_t x_t,\quad a\subset t=u_1\cap u_2,\\[.3ex]
     \text{then} & d_a^{(u_2)}|_t = \beta_t \partial_t - \gamma_t x_t.
 \end{array}
\end{equation}
This can certainly be done because of~\eqref{0} and the freedom \eqref{renorm} and~\eqref{interc}.

In order to ensure~\eqref{agree-t} for all three~$t$ in either l.h.s.\ or r.h.s., it remains to prove the following Lemma.

\begin{lemma}\label{l:ci}
The three consistency conditions~\eqref{0} (corresponding to three inner tetrahedra) for common factors at scalar products for either l.h.s.\ or r.h.s.\ are also consistent between themselves. 
\end{lemma}

\begin{proof}
Consider, for instance, the l.h.s. It is enough to consider the scalar product of $d_{12}$ and~$d_{13}$. It must vanish for any 4-simplex:
\[
\langle d_{12}^{(12345)}, d_{13}^{(12345)} \rangle = \langle d_{12}^{(12346)}, d_{13}^{(12346)} \rangle = \langle d_{12}^{(12356)}, d_{13}^{(12356)} \rangle = 0 .
\]
For the 4-simplex~$12345$, it consists of two parts corresponding to tetrahedra $1234$ and~$1235$, so
\[
\langle d_{12}^{(12345)}, d_{13}^{(12345)} \rangle_{1234} = - \langle d_{12}^{(12345)}, d_{13}^{(12345)} \rangle_{1235}.
\]
Then, this scalar product changes its sign when passing to the 4-simplex~$12346$:
\[
\langle d_{12}^{(12346)}, d_{13}^{(12346)} \rangle_{1234} = - \langle d_{12}^{(12345)}, d_{13}^{(12345)} \rangle_{1234},
\]
according to~\eqref{0}. This way, the sign changes six times while walking around the 2-face~$123$, and thus the factor at the scalar product returns to its good old value.
\end{proof}

In the following Lemma~\ref{l:e}, we consider both sides of move 3--3 (or rather relation~\eqref{33}) separately, not knowing yet whether their subspaces~$\mathfrak V$ coincide.

\begin{lemma}\label{l:e}
If the edge operators for separate 4-simplices are built from a 2-cocycle~$\Omega$ on~$\partial\Delta^5$ and further chosen according to conditions~\eqref{agree-t}, then the maximal isotropic subspace~$\mathfrak V$ for either l.h.s.\ or r.h.s.\ of the move 3--3 is spanned by compositions of edge operators constructed according to Lemma~\ref{l:lhs33}: for an edge~$a$ and the three 4-simplices~$u$, take the three~$d_a^{(u)}$ as $d'$, $d''$ and~$d'''$.
\end{lemma}

\begin{proof}
The composability of edge operators follows from comparing \eqref{agree-t} with~\eqref{agree-d}, and that the compositions give the whole nine-dimensional~$\mathfrak V$ follows from a direct calculation.
\end{proof}

\subsubsection*{Consistency conditions for the boundary tetrahedra}

Here common factors at scalar products for the two 4-simplices adjacent to the same boundary tetrahedron~$t$ --- one of these in the l.h.s.\ and one in the r.h.s., we call them $u_{\mathrm{l.h.s.}}$ and~$u_{\mathrm{r.h.s.}}$ --- must be chosen so as to be identical on~$t$:
\[
\langle d_a^{(u_{\mathrm{l.h.s.})}}, d_b^{(u_{\mathrm{l.h.s.})}} \rangle = \langle d_a^{(u_{\mathrm{r.h.s.})}}, d_b^{(u_{\mathrm{r.h.s.})}} \rangle, \qquad a,b\subset t=u_{\mathrm{l.h.s.}}\cap u_{\mathrm{r.h.s.}}.
\]
Then, the Grassmann variables~$x_t$ for boundary~$t$ must be chosen so that
\begin{equation}\label{agree-b}
d_a^{(u_{\mathrm{l.h.s.}})} = d_a^{(u_{\mathrm{r.h.s.}})}.
\end{equation}

\begin{lemma}\label{l:a}
The consistency conditions~\eqref{agree-b} for the boundary tetrahedra are also consistent between themselves and with the conditions~\eqref{agree-t} for inner tetrahedra.
\end{lemma}

\begin{proof}
Take \emph{any} 2-face~$s$ in~$\partial\Delta^5$ --- the union of l.h.s.\ and r.h.s.\ of the move 3--3 --- and walk around it similarly to what we did in the proof of Lemma~\ref{l:ci}, counting how many times the partial product $\langle d_a^{(u)}, d_b^{(u)} \rangle_t$ of edge operators for two different edges $a,b\subset s$ changes its sign when we change either $t$ or~$u$. The number of times is always even: 3 when changing the~$t$'s plus 1 or 3 when changing the~$u$'s.
\end{proof}

\begin{theorem}\label{th:all}
The edge operators for separate 4-simplices built from a generic 2-cocycle~$\Omega$ on~$\partial\Delta^5$ can further be chosen according to conditions~\eqref{agree-t} and~\eqref{agree-b}, and in this case, the relation~\eqref{33} holds.
\end{theorem}

\begin{proof}
It remains to note that conditions~\eqref{agree-b} ensure that the edge operators in~$\mathfrak V$ --- compositions of 4-simplex edge operators --- are the same for every edge in the common boundary of both sides of the move 3--3.
\end{proof}

\subsection{The 18-parameter family}\label{ss:18p}

Our 2-cocycles~$\Omega$ on~$\partial\Delta^5$, taken up to a factor, make a 9-parametric family. Renormalizations~\eqref{renorm} for Grassmann variables on boundary 3-faces (table~\eqref{table}) supply 9 more parameters, and all these 18 parameters are independent. 

On the other hand, every quasi-Gaussian weight~$\mathcal W_{ijklm}$ in~\eqref{33} depends on 10 parameters (up to a factor. If, for instance, $\mathcal W_{ijklm}$~is Gaussian, the parameters are entries of matrix~\eqref{F}). When we compose the l.h.s.\ or r.h.s.\ of~\eqref{33} (not yet demanding that l.h.s.\ be equal to r.h.s., and not taking the value $\const$ into account),
there are thus 30 parameters. Three of them are, however, redundant, because of the possible scalings~\eqref{renorm} of variables~$x_t$ on three \emph{inner} tetrahedra. So, we have $3\times 10 - 3 = 27$ essential parameters in each side of~\eqref{33}.

The l.h.s.\ or r.h.s.\ of~\eqref{33} is determined (up to a factor) by a 9-dimensional isotropic subspace in an 18-dimensional complex Euclidean space, and the space (isotropic Grassmannian) of such subspaces is 36-dimensional. So, requiring the equalness of these subspaces for the l.h.s.\ and r.h.s., we subtract 36 parameters and are left with $2\times 27 - 36 = 18$ parameters.

\begin{remark}
We have practically repeated the reasoning in~\cite[Subsection~5.1]{KS2}, where it was called heuristic. To make it rigorous, we must check that that the 36 conditions appearing in the previous paragraph are independent, which can be done on a computer by checking that a relevant $36\times 36$ Jacobian determinant does not vanish identically.
\end{remark}

The conclusion is that a \emph{Zariski open subset of all relations~\eqref{33} comes out the way we have described in this paper}.

\begin{remark}
Particular or limit cases of our ``general position'' relations may, nevertheless, be of their own interest, and require separate investigation. For instance, it can be shown that our Lemma~\ref{l:w} no longer holds for the edge operators corresponding to the family of weights presented in~\cite[Subsection~6.3]{KS2}. Recall that intriguing exotic homological nature of (some parameters in) that family was revealed in~\cite[Section~8]{KS2}.
\end{remark}

\subsection*{Acknowledgments}

I thank the referee for the comments that helped me to improve this paper.

\end{document}